\documentclass[11pt]{scrartcl}

\usepackage{url}
\usepackage[hidelinks]{hyperref}
\usepackage[utf8]{inputenc}
\usepackage[small]{caption}
\usepackage{graphicx}
\usepackage{amsmath}
\usepackage{booktabs}
\usepackage{amsthm}
\usepackage{amssymb}

\usepackage{cleveref}
\usepackage{color}
\usepackage{xspace}
\usepackage{tikz}
\usepackage{tabularx}
\usepackage{pbox}
\usepackage{framed}
\usepackage[shortlabels]{enumitem}

\usepackage{natbib}

\newtheorem{theorem}{Theorem}[section]
\newtheorem{corollary}[theorem]{Corollary}
\newtheorem{proposition}[theorem]{Proposition}
\newtheorem{lemma}[theorem]{Lemma}

\theoremstyle{definition}
\newtheorem{definition}[theorem]{Definition}

\newtheorem*{claim}{Claim}

\newcommand{\N}{\mathbb{N}}
\newcommand{\bA}{\mathbf{A}}
\newcommand{\bB}{\mathbf{B}}
\newcommand{\bD}{\mathbf{D}}

\newcommand{\bzero}{\mathbf{0}}
\newcommand{\bone}{\mathbf{1}}
\newcommand{\R}{\mathbb{R}}
\newcommand{\bw}{\mathbf{w}}

\newcommand{\bv}{\mathbf{v}}
\newcommand{\bx}{\mathbf{x}}
\newcommand{\by}{\mathbf{y}}
\newcommand{\bz}{\mathbf{z}}
\newcommand{\bH}{\mathbf{H}}
\newcommand{\bW}{\mathbf{W}}

\newcommand{\bC}{\mathbf{C}}
\newcommand{\cprop}{c^{\text{PROP}}}
\newcommand{\cef}{c^{\text{EF}}}
\newcommand{\ccdiv}{c^{\text{CD}}_k}
\newcommand{\cpropemph}{c^{\emph{PROP}}}
\newcommand{\cefemph}{c^{\emph{EF}}}
\newcommand{\ccdivemph}{c^{\emph{CD}}_k}
\newcommand{\setsplit}{\textsc{Max-2-2-Set-Splitting}}
\newcommand{\eps}{\varepsilon}
\DeclareMathOperator{\poly}{poly}

\DeclareMathOperator{\exdisc}{disc^{max}}

\DeclareMathOperator{\disc}{disc}
\DeclareMathOperator{\wdisc}{wdisc}
\newcommand{\exwdisc}{\operatorname{wdisc}^{\operatorname{max}}}

\allowdisplaybreaks

\title{Almost Envy-Freeness for Groups: \\ Improved Bounds via Discrepancy Theory}

\author{
Pasin Manurangsi\\Google Research
\and
Warut Suksompong\\National University of Singapore
}

\date{\vspace{-5ex}}

\begin{document}

\maketitle

\begin{abstract}
We study the allocation of indivisible goods among groups of agents using well-known fairness notions such as envy-freeness and proportionality. While these notions cannot always be satisfied, we provide several bounds on the optimal relaxations that can be guaranteed. For instance, our bounds imply that when the number of groups is constant and the $n$ agents are divided into groups arbitrarily, there exists an allocation that is envy-free up to $\Theta(\sqrt{n})$ goods, and this bound is tight. Moreover, we show that while such an allocation can be found efficiently, it is NP-hard to compute an allocation that is envy-free up to $o(\sqrt{n})$ goods even when a fully envy-free allocation exists. Our proofs make extensive use of tools from discrepancy theory.
\end{abstract}  

\section{Introduction}

Resource allocation problems arise in numerous facets of modern society, from allotting supplies to neighborhoods in a city to distributing personnel among governmental organizations.
A principal consideration when allocating resources is \emph{fairness}: the society is better off when all parties involved feel that they receive a fair share of the resource.
It therefore comes as no surprise that the study of how to allocate resources fairly---commonly referred to as \emph{fair division}---has received substantial attention in economics and, as societies become more interconnected and applications grow in scale, in computer science \citep{BramsTa96,Moulin03,Moulin19,Thomson16}

The vast majority of the fair division literature assumes that each involved party consists of a single agent.
Yet, in many resource allocation scenarios, especially large-scale ones, resources are allocated to \emph{groups} of agents---even though the agents in each group share the same set of goods, they may have varying preferences over different goods in the set.
Indeed, some citizens of a neighborhood may benefit from new books allotted to the public library, while others would rather have additional fitness equipment in their local park.
Similarly, members of an organization may have diverse opinions about the new personnel that they would like to have in their organization.
These scenarios cannot be captured by the traditional fair division setting, in which each recipient of a bundle of goods is represented by a single preference.

The group aspect of fair division has been addressed in a number of recent papers \citep{ManurangsiSu17,GhodsiLaMo18,Suksompong18,Suksompong18-2,SegalhaleviNi19,SegalhaleviSu19,SegalhaleviSu21,SegalhaleviSu22,KyropoulouSuVo20}.
Most of these papers studied the important fairness notion of \emph{envy-freeness}: an agent is said to be \emph{envy-free} if she values the goods allocated to her group at least as much as those allocated to any other group.
When goods are discrete---books, personnel, fitness equipment, and many other common supplies fall into this category---envy-freeness cannot always be satisfied even when allocating the goods among individual agents; indeed, this can be easily seen when there is a single valuable good and at least two agents.
This observation has motivated relaxing the envy-freeness criterion to \emph{envy-freeness up to $c$ goods (EF$c$)}, which means that any agent's envy toward another group can be eliminated by removing at most $c$ goods from that group's bundle, where $c\ge 1$ is an integer parameter.

When allocating goods among individual agents, an EF1 allocation can be found regardless of the number of agents \citep{LiptonMaMo04}.
However, the picture for group allocation is much less clear, even in the simplest case of two groups.
\citet{SegalhaleviSu19} showed that if the two groups contain $n$ agents in total and the agents have additive valuations, then an EF$n$ allocation is guaranteed to exist.
Their result follows by applying a classic theorem on \emph{consensus halving}, i.e., a partition of a set of \emph{divisible} goods into two parts such that every agent values both parts equally.
Since there is always a consensus halving in which at most $n$ goods are divided \citep{Alon87,SimmonsSu03}, rounding such a consensus halving yields an EF$n$ allocation.\footnote{In fact, \citet{SegalhaleviSu19} gave a slightly better guarantee of EF$(n-1)$. This guarantee was obtained by finding a consensus halving for $n-1$ of the agents, and letting the remaining agent choose the part that she prefers. }
On the other hand, \citet[Prop.~3.5]{KyropoulouSuVo20} gave a simple example showing that it is impossible to ensure EF$c$ for $c\in o(\log n)$, thereby leaving an exponential gap in this fundamental question.
Can we always achieve an impressive fairness guarantee of mere logarithmic envy, or does the envy scale linearly with the number of agents in the worst case?

\subsection{Our Results}

In this paper, we give a precise answer to the above question, and much more.
We consider a general setting with $n=n_1+\dots+n_k$ agents distributed into $k\ge 2$ groups consisting of $n_1,\dots,n_k\ge 1$ agents, respectively.
As is common in fair division, we assume that the agents have additive utilities over the goods.
Besides EF$c$, we investigate relaxations of two other important fairness notions: \emph{proportionality}---every agent believes that the share allocated to her group is worth at least $1/k$ of the entire set of goods---and \emph{consensus $1/k$-division}---each agent finds all $k$ bundles to be of equal value.\footnote{When $k=2$, consensus $1/k$-division is better known as \emph{consensus halving} \citep{SimmonsSu03}.}
The precise definitions can be found in \textbf{\Cref{sec:fairness-notions}}.
We remark here that even though the definition of consensus $1/k$-division does not involve groups, the notion is natural and has been studied in several papers (see \Cref{sec:related}), and will moreover be useful as an intermediate notion through which we can obtain results for the other notions of interest. 

For each fairness notion and each $n_1,\dots,n_k$, we are interested in the smallest positive value $c$ such that an allocation satisfying that notion up to $c$ goods always exists for agents with arbitrary additive utilities.
For envy-freeness and proportionality, we denote this value of $c$ by $\cef(n_1, \dots, n_k)$ and $\cprop(n_1, \dots, n_k)$, respectively.
On the other hand, for consensus $1/k$-division, the partition of agents into groups is inconsequential, so we use the notation $\ccdiv(n)$.
Our main results provide bounds on these values:

\begin{theorem} \label{thm:main-ef}
For any $k, n_1, \dots, n_k \in \N$, 
\begin{align*}
O(\sqrt{n}) 
&\geq \cefemph(n_1, \dots, n_k) \geq \Omega(\sqrt{\max\{n_1, \dots, n_k\}/k^3}).    
\end{align*}
\end{theorem}

\begin{theorem} \label{thm:main-prop}
For any $k, n_1, \dots, n_k \in \N$, 
\begin{align*}
O(\sqrt{n}) 
&\geq \cpropemph(n_1, \dots, n_k) \geq \Omega(\sqrt{\max\{n_1, \dots, n_k\}/k^3}).    
\end{align*}
\end{theorem}

\begin{theorem} \label{thm:main-cd}
For any $n, k \in \N$, 
\[O(\sqrt{n}) \geq \ccdivemph(n) \geq \Omega(\sqrt{n/k}).\]
\end{theorem}

Note that since $\max\{n_1, \dots, n_k\} \ge n/k$, all three bounds are asymptotically tight when $k$ is constant.
In particular, \Cref{thm:main-ef} answers the question that we posed earlier: taking $k=2$, we find that the optimal envy-freeness guarantee for two groups is EF$c$ where $c\in\Theta(\sqrt{n})$.
This significantly improves upon the lower bound of $\Omega(\log n)$ and upper bound of $O(n)$ from prior work, and implies that a decent, though not outstanding, fairness guarantee can be obtained.
We establish \Cref{thm:main-cd} along with the upper bounds of \Cref{thm:main-ef,thm:main-prop} in \textbf{\Cref{sec:approx-multi}}, and the lower bounds of \Cref{thm:main-ef,thm:main-prop} in \textbf{\Cref{sec:lower-weighted}}.

Our main tools and techniques throughout this work come from \emph{discrepancy theory}, an area of mathematics that studies how much deviation from the desired state is necessary in various settings---we provide the relevant background in \textbf{\Cref{sec:discrepancy}}.
The tools that we use imply that for each of the fairness notions, an allocation satisfying the corresponding upper bound in Theorems~\ref{thm:main-ef}--\ref{thm:main-cd} can be found efficiently.
In light of this, a natural question is whether we can compute an allocation improving upon these bounds \emph{if} such allocations are known to exist in a given instance.
In \textbf{\Cref{sec:complexity}}, we provide a strong negative answer to this question: for example, we show that even if a fully envy-free allocation is known to exist for a certain instance, it is still NP-hard to find an allocation that is envy-free up to $o(\sqrt{n})$ goods for that instance.

\subsection{Further Related Work}
\label{sec:related}

While fair division has a long and storied history, several fairness notions for the indivisible goods setting, including envy-freeness relaxations, have only been proposed and studied in the past few years \citep{BouveretChMa16,Markakis17,Suksompong21,AmanatidisBiFi22,AzizLiMo22}.
In the group setting, \citet{KyropoulouSuVo20} showed that EF1 can be guaranteed for all agents only when the groups are small---for instance, with two groups, an EF1 allocation does not always exist when both groups have size at least three.
\citet{SegalhaleviSu19} investigated \emph{democratic fairness}, where the goal is to satisfy a certain fraction of the agents in each group.
They showed that for two groups with any number of agents, there exists an allocation that is EF1 for at least half of the agents in each group---this ratio is tight in the worst case, and continues to be tight even if we relax EF1 to EF$c$ for any constant $c$.

Besides the model that we consider, a number of papers have studied related models and notions.
\citet{GhodsiLaMo18} addressed \emph{rent division} among groups, where in addition to deciding the allocation of the rooms, the agents must determine how to split the rent of their apartment.
\citet{BenabbouChEl19} examined a group setting where the goods allocated to each group are further divided among the members of the group, so in contrast to our setting, each agent does not derive full utility from the bundle of her group.
Several authors studied individual resource allocation using fairness notions relating different groups of agents, for example notions aiming to minimize envy that arises between groups \citep{Berliant92,Husseinov11,TodoLiHu11,AleksandrovWa18,ConitzerFrSh19,AzizRe20}. 

Like fair division in general, consensus $1/k$-division and consensus halving have been studied by mathematicians and economists for several decades \citep{HobbyRi65,Alon87,SimmonsSu03}, and attracted recent interest from computer scientists in light of new computational complexity results \citep{FilosratsikasGo18,FilosratsikasGo22,FilosratsikasFrGo18,FilosratsikasHoSo20,DeligkasFeMe21,DeligkasFiHo21,DeligkasFeHo22,GoldbergHoIg22}.
In particular, \citet{FilosratsikasGo18} proved that approximate consensus halving of a one-dimensional heterogeneous divisible resource is PPA-complete---this constituted the first PPA-completeness result for a problem that is ``natural'' in the sense that its description does not involve a polynomial-sized circuit.

\section{Preliminaries}

Let $G=[m]$ be the set of goods, where $[r] := \{1,2,\dots,r\}$ for any positive integer $r$.
There are $n = n_1+\dots+n_k$ agents divided into $k\ge 2$ groups, where group~$i$ contains $n_i\ge 1$ agents.
Denote by $a^{(i, j)}$ the $j$th agent in group~$i$. 
The utility of $a^{(i, j)}$ for good~$\ell$ is given by $u^{(i, j)}(\ell)$.
We assume that the agents' utilities are additive, that is, $u^{(i, j)}(G') = \sum_{\ell\in G'}u^{(i, j)}(\ell)$ for every $G'\subseteq G$.
An \emph{allocation} $(A_1,\dots,A_k)$ is an ordered partition of the goods into $k$ bundles, where bundle $A_i$ is allocated to group~$i$.
In particular, $A_1\cup\dots\cup A_k = G$ and $A_i\cap A_j = \emptyset$ for $i\ne j$.

\subsection{Fairness Notions}
\label{sec:fairness-notions}

We are interested in the following fairness notions:

\begin{definition}
Let $c$ be a nonnegative integer.
An allocation $(A_1,\dots,A_k)$ is said to be 
\begin{itemize}
\item \emph{envy-free up to $c$ goods (EF$c$)} if, for every agent $a^{(i, j)}$ and every $i'\ne i$, there exists a set $B\subseteq A_{i'}$ with $|B|\le c$ such that $u^{(i, j)}(A_i) \ge u^{(i, j)}(A_{i'}\setminus B)$.
\item \emph{proportional up to $c$ goods (PROP$c$)} if, for every agent $a^{(i, j)}$, there exists a set $B\subseteq G\setminus A_i$ with $|B|\le c$ such that $u^{(i, j)}(A_i) \ge u^{(i, j)}(G)/k - u^{(i, j)}(B)$.
\item a \emph{consensus $1/k$-division up to $c$ goods} if, for every agent $a$ and every pair of bundles $A_i,A_{i'}$, there exists $B\subseteq A_{i'}$ with $|B|\le c$ such that $a$ values $A_i$ no less than $A_{i'} \setminus B$.
\end{itemize}
\end{definition}

Note that unlike the first two notions, the third notion does not depend on how the agents are distributed across groups.
EF$c$ has been studied in many fair division papers, PROP$c$ has also been considered in a number of papers \citep{BarmanKr19,AzizMoSa20,ChakrabortyScSu21}, whereas consensus $1/k$-division up to $c$ goods is proposed for the first time in this paper, to the best of our knowledge.

For each $k,n_1,\dots,n_k$, let $\cef(n_1, \dots, n_k)$ (resp., $\cprop(n_1, \dots, n_k)$) denote the smallest value of $c$ such that an EF$c$ (resp., PROP$c$) allocation is guaranteed to exist for agents with additive utilities.
Similarly, let $\ccdiv(n)$ denote the analogous value for consensus $1/k$-division up to $c$ goods when there are $n$ agents and $k$ bundles.
We have the following relations between these values.

\begin{proposition}
\label{prop:fairness-relation}
For any $k,n',n_1,\dots,n_k\in\mathbb{N}$, we have
\begin{enumerate}[(a)]
\item $\cefemph(n_1, \dots, n_k) \leq \ccdivemph(n_1 + \cdots + n_k)$;
\item $\cpropemph(n_1, \dots, n_k) \leq \cefemph(n_1, \dots, n_k)$;
\item $\cefemph(n', \dots, n') \geq \ccdivemph(n')$, where there are $k$ copies of $n'$ on the left-hand side.
\end{enumerate}
\end{proposition}

\begin{proof}
We prove the three relations in turn.

\begin{enumerate}[(a)]
\item This follows immediately from the observation that a consensus $1/k$-division up to $c$ goods for $n_1+\dots+n_k$ agents is also envy-free up to $c$ goods for these agents regardless of how the agents are distributed into groups.

\item It suffices to show that for any $c$, every EF$c$ allocation is also PROP$c$.
Let $(A_1, \dots, A_k)$ be an EF$c$ allocation, and consider agent~$a^{(i, j)}$.
By definition of EF$c$, for each $i'\ne i$, there exists $B_{i'}\subseteq A_{i'}$ with $|B_{i'}|\le c$ such that $u^{(i, j)}(A_i) \geq u^{(i, j)}(A_{i'} \setminus B_{i'})$. 

Let $B$ denote the set of the $c$ most valuable goods for $a^{(i, j)}$ outside of $A_i$, breaking ties arbitrarily.
We have $u^{(i, j)}(B) \geq u^{(i, j)}(B_{i'})$ for all $i' \ne i$. 
Letting $B_i = \emptyset$, we have
\begin{align*}
u^{(i, j)}(A_i) &\geq \frac{1}{k} \sum_{i' \in [k]}u^{(i, j)}(A_{i'} \setminus B_{i'}) \\
&\geq \frac{1}{k} \sum_{i' \in [k]} \left(u^{(i, j)}(A_{i'}) - u^{(i, j)}(B)\right) = \frac{u^{(i, j)}(G)}{k} - u^{(i, j)}(B).
\end{align*}
Hence, the allocation is PROP$c$, as desired.

\item Given $n'$ agents, we make a copy of each agent in each of the $k$ groups.
Any EF$c$ allocation with respect to these groups is also a consensus $1/k$-division for the original agents.
The conclusion follows. \qedhere
\end{enumerate}
\end{proof}

\subsection{Discrepancy Theory}
\label{sec:discrepancy}

In this section, we outline the tools from discrepancy theory that we will use in this work.
Intuitively, a basic connection between discrepancy theory and our group fair division setting is the following:
Discrepancy theory considers a setting where there is a collection of subsets (also known as a \emph{set system}) and we want to color the elements of the ground set in two colors so that each subset contains roughly the same number of elements of each color.
The elements of the ground set correspond to the goods in our setting, while each subset represents an agent and its elements correspond to the goods that the agent values.
The goal of discrepancy theory is therefore similar to that of dividing the goods into two sets so that each agent values the two sets almost equally.

In this work, we view a vector $\bv \in \R^k$ also as a column matrix $\bv \in \R^{k \times 1}$. 
We use $\bone_k$ to denote the $k$-dimensional all-$1$ vector; when the dimension is clear from context, we may drop the subscript and simply write $\bone$. Furthermore, for a set $S \subseteq [k]$, we write $\bone(S) \in \{0, 1\}^k$ to denote the indicator vector of $S$, i.e., $(\bone(S))_i = 1$ if and only if $i \in S$.

For every $p \in [1, \infty)$, we use $\|\bv\|_p$ to denote the $\ell_p$ norm of $\bv$, defined by $\left(\sum_{i \in [k]} |\bv_i|^p\right)^{1/p}$. The $\ell_\infty$ norm of $\bv$, denoted by $\|\bv\|_{\infty}$, is $\max_{i \in [k]} |\bv_i|$. 
Further, denote by $\|\bv\|_0$ the number of nonzero coordinates of $\bv$.

\subsubsection{2-Color Discrepancy}

As mentioned earlier, a classic scenario in discrepancy theory is when there is a set system and the goal is to color the elements in two colors in such a way that each subset contains roughly the same number of elements of each color. 
As \cite{LovaszSpVe86} noted, this notion generalizes naturally to any matrix. Specifically, the \emph{discrepancy} of a matrix $\bA \in \R^{n \times m}$ is defined as
\begin{align*}
\disc(\bA) := \min_{\bx \in \{0, 1\}^m} \|\bA(0.5 \cdot \bone - \bx)\|_{\infty}.
\end{align*}
Here $\bx$ can be thought of as a $2$-coloring and the quantity $\|\bA(0.5 \cdot \bone - \bx)\|_{\infty}$ measures how ``unbalanced'' it is.

Let
\begin{align*}
\exdisc(n) := \sup_{m \in \N} \sup_{\bA \in [0, 1]^{n \times m}} \disc(\bA).
\end{align*}

A priori, $\exdisc(n)$ might even be infinite since we allow $\bA$ to have an arbitrary number of columns. Remarkably, however, it is known\footnote{We note here that some of the bounds we refer to are stated only for $0$-$1$ matrices $\bA$. However, one can check that they also hold for any matrix $\bA \in [0, 1]^{n \times m}$. See also Appendix~\ref{app:discrepancy-algo} where we sketch how the constructive versions of these bounds can be derived.}\textsuperscript{,}\footnote{See also~\citep{Spencer85} on which this bound is based.} that $\exdisc(n)$ is bounded by $O(\sqrt{n})$: 

\begin{lemma}[{\cite[Corollary 12.3.4]{AlonSp00}}] \label{lem:disc-ub}
For any $n \in \N$, $\exdisc(n) \leq O(\sqrt{n})$.
\end{lemma}

The above bound is also known to be asymptotically tight:

\begin{lemma}[\citep{Spencer85}] \label{lem:disc-lb}
For any $n \in \N$, it holds that $\exdisc(n) \geq \Omega(\sqrt{n})$.
\end{lemma}

\paragraph{Weighted Discrepancy.} The \emph{$p$-weighted discrepancy}~\citep{DoerrSr03} is a generalization of discrepancy where $0.5$ is replaced by some $p \in [0, 1]$:
\begin{align*}
\wdisc_p(\bA) := \min_{\bx \in \{0, 1\}^m} \|\bA(p \cdot \bone - \bx)\|_{\infty}.
\end{align*}
Similarly to above, let
\begin{align*}
\exwdisc_p(n) := \sup_{m \in \N} \sup_{\bA \in [0, 1]^{n \times m}} \wdisc_p(\bA).
\end{align*}

Using standard techniques in discrepancy theory, we can prove the following lower bound on $\exwdisc_p(n)$. (This bound was also implicit in the work of Doerr and Srivastav.)

\begin{proposition} \label{prop:wdisc-lb}
For any $p \in (0, 1/2]$ and any $n \in \N$ such that $n \geq 16/p$, we have $\exwdisc_p(n) \geq \Omega(\sqrt{pn})$.
\end{proposition}

To prove \Cref{prop:wdisc-lb}, we need the following lemma, which was implicit in the work of \cite{Chazelle01} and made explicit by \cite{CharikarNeNi11}.
Recall that a Hadamard matrix of order $n$ is a matrix $\bH \in \{-1, +1\}^{n\times n}$ such that the rows are mutually orthogonal. 
By multiplying $-1$ to each column as necessary, we may assume throughout that the first row of $\bH$ consists of only $1$'s. Similarly, we may assume that the first column of $\bH$ consists of only $1$'s.

\begin{lemma}[\citep{Chazelle01,CharikarNeNi11}] \label{lem:hadamard-discrepancy}
Let $\bH$ be a Hadamard matrix of order $n$, and $\bW := \frac{1}{2}(\bone_{n \times n} + \bH)$ where $\bone_{n \times n}$ is an $n \times n$ all-$1$ matrix. 
Then, for any $\bz \in \R^n$, we have $\|\bW\bz\|_2^2 \geq n \cdot \left(\sum_{i \in  [n] \setminus \{1\}} \bz_i^2\right)$.
\end{lemma}

\begin{proof}[Proof of \Cref{prop:wdisc-lb}]
We will show that $\exwdisc_p(n) \geq \Delta := \sqrt{pn/8}$ for every $n \geq 16/p$ such that a Hadamard matrix of order $n$ exists. 
Since it is known that such a matrix exists whenever $n$ is a power of two~\citep{sylvester67} and the function $\exwdisc_p(n)$ is nondecreasing in $n$, this also establishes the lower bound of the same order for all $n$.

We choose $\bA = \bW$ defined in \Cref{lem:hadamard-discrepancy}. Now, consider any $\bx \in \{0, 1\}^n$. Since the first row of $\bA$ is an all-$1$ vector, we have $\|\bA(p \cdot \bone - \bx)\|_{\infty} \geq \left|\left<\bone, p \cdot \bone - \bx\right>\right| = |pn - \left<\bone, \bx\right>|$.
Thus, if $|pn - \left<\bone, \bx\right>| \geq \Delta$, then we have $\|\bA(p \cdot \bone - \bx)\|_{\infty} \geq \Delta$, as desired. Therefore, it suffices to consider the case $|pn - \left<\bone, \bx\right>| < \Delta$. This means that more than $pn - \Delta$ coordinates of $\bx$ are equal to $1$; for such a coordinate $i$, we have $(p \cdot \bone - \bx)_i = -(1 - p)$. This implies that
\begin{align*}
\|p \cdot \bone - \bx\|_2^2 &\geq (1 - p)^2 \cdot (pn - \Delta) \ge \frac{pn - \Delta}{4} \geq \Delta^2 + 1;
\end{align*}
the last inequality is equivalent to $pn \ge \sqrt{pn/2} + 8$ and follows from our assumption that $n \geq 16/p$.

Hence, applying \Cref{lem:hadamard-discrepancy}, we arrive at
\begin{align*}
\|\bA(p \cdot \bone - \bx)\|_2^2 \geq n\left(\|p \cdot \bone - \bx\|_2^2 - 1\right) \geq n\Delta^2,
\end{align*}
which implies that $\|\bA(p \cdot \bone - \bx)\|_\infty \geq \Delta$, as desired.
\end{proof}

\subsubsection{Multi-Color Discrepancy}

We will also use the extension of the $2$-color definition to multi-color cases due to \citet{DoerrSr03}. Recall that a \emph{$k$-coloring} of $[m]$ is a function $\chi: [m] \to [k]$. The \emph{$k$-color discrepancy} of $\bA$ is defined as
\begin{align*}
\disc(\bA, k) := \min_{\chi: [m] \to [k]} \max_{s \in [k]} \left\|\bA\left(\frac{1}{k} \cdot \bone - \bone(\chi^{-1}(s))\right)\right\|_{\infty}.
\end{align*}
Note that $\disc(\bA, 2)$ coincides with $\disc(\bA)$ defined earlier.

Similarly to above, we let
\begin{align*}
\exdisc(n, k) := \sup_{m \in \N} \sup_{\bA \in [0, 1]^{n \times m}} \disc(\bA, k).
\end{align*}

The following lemma is a consequence of Corollary 3.5 of~\cite{DoerrSr03} and \Cref{lem:disc-ub}. 
\begin{lemma}[\citep{DoerrSr03}] \label{lem:disc-ub-multi}
For any $n, k \in \N$, $\exdisc(n, k) \leq O(\sqrt{n})$.
\end{lemma}

Furthermore, Doerr and Srivastav also proved the following lower bound:\footnote{\citet{DoerrSr03} only stated their lower bound for $n$ such that a Hadamard matrix of order~$n$ exists. However, as explained in our proof of \Cref{prop:wdisc-lb}, this implies the same asymptotic bound for all $n \in \N$.}

\begin{lemma}[{\cite[Theorem 5.2]{DoerrSr03}}] \label{lem:disc-lb-multi}
For any $n,k \in \N$ such that $k \geq 2$, $\exdisc(n, k) \geq \Omega(\sqrt{n/k})$.
\end{lemma}

\section{Approximate Fair Division From Multi-Color Discrepancy}
\label{sec:approx-multi}

In this section, we derive generic upper and lower bounds for the value $\ccdiv$ based on the multi-color discrepancy bounds $\exdisc(n, k)$. Our results are stated formally below.
\begin{theorem} \label{thm:lb-direct}
For any $n, k \in \N$, we have
\[\ccdivemph(n) \geq \lceil \exdisc(n, k) \rceil.\]
\end{theorem}

\begin{theorem} \label{thm:ub-main}
For any $n, k \in \N$, we have \[\ccdivemph(n) \leq 4 \cdot \lceil \exdisc(2n, k) \rceil.\]
\end{theorem}

By the known bounds for $\exdisc$ (\Cref{lem:disc-ub-multi,lem:disc-lb-multi}), the two results above yield \Cref{thm:main-cd}.

Using the relationships between $\ccdiv$, $\cef$, and $\cprop$  established in \Cref{prop:fairness-relation}, we get the following corollary:

\begin{corollary}
For any $k, n', n_1, \dots, n_k \in \N$, we have
\begin{enumerate}[(a)]
\item $\cefemph(n_1, \dots, n_k) \leq 4 \cdot \lceil \exdisc(2(n_1 + \dots + n_k), k) \rceil$;
\item $\cpropemph(n_1, \dots, n_k) \leq 4 \cdot \lceil \exdisc(2(n_1 + \dots + n_k), k) \rceil$;
\item $\cefemph(n', \dots, n') \geq \lceil \exdisc(n', k) \rceil$, where there are $k$ copies of $n'$ on the left-hand side.
\end{enumerate}
\end{corollary}

Parts (a) and (b), together with the known upper bound for $\exdisc$ (\Cref{lem:disc-ub-multi}), give us the upper bounds in \Cref{thm:main-ef,thm:main-prop}.
On the other hand, part (c) is not sufficient for the lower bounds in these theorems yet, since the above corollary does not apply to the ``unbalanced groups'' case where some groups are small, e.g., when $n_1 = \cdots = n_{k - 1} = 1$ and $n_k = n'$. Indeed, the lower bound parts of \Cref{thm:main-ef,thm:main-prop}  will be handled in \Cref{sec:lower-weighted}.

\subsection{Lower Bound}

We prove \Cref{thm:lb-direct} via a simple reduction that views each row of a matrix $\bA$ as a vector of utilities for the goods. 
The existence of a consensus $1/k$-division up to a small number of goods would imply a strong upper bound on the discrepancy of $\bA$. 
This is formalized below; since the distribution of agents into groups is irrelevant for consensus $1/k$-division, we use the notation $a^j$ and $u^j$ instead of $a^{(i,j)}$ and $u^{(i,j)}$.

\begin{proof}[Proof of~\Cref{thm:lb-direct}]
Let $\Delta := \lceil \exdisc(n, k) \rceil$.
Note that $\Delta - 1 < \exdisc(n, k)$. Thus, there exists $m\in\N$ and $\bA \in [0, 1]^{n \times m}$ such that $\disc(\bA,k) > \Delta - 1$. We define the agents' utilities by $u^{j}(\ell) = \bA_{j, \ell}$ for all $j \in [n]$ and $\ell \in [m]$.
It suffices to show that there is no consensus $1/k$-division up to $\Delta - 1$ goods with respect to these utilities.

Suppose for the sake of contradiction that there is a consensus $1/k$-division up to $\Delta - 1$ goods, $(A_1, \dots, A_k)$.
Let $\chi: [m] \to [k]$ denote the coloring corresponding to this allocation, i.e., $\chi^{-1}(s) = A_s$.

Consider any agent $a^j$. Since the allocation is a consensus $1/k$-division up to $\Delta - 1$ goods and each good has value at most $1$ to each agent, the following holds for all $i,s \in [k]$:
\begin{align} \label{eq:bounded-diff}
\Delta - 1 &\geq |u^j(A_i) - u^j(A_s)|.
\end{align}
Fix any $s\in[k]$.
For each $j\in[n]$, this inequality allows us to bound the $j$th entry of $\bA\left(\frac{1}{k}\cdot\bone - \bone(\chi^{-1}(s))\right)$ as
\begin{align*}
\left|\left(\bA\left(\frac{1}{k}\cdot\bone - \bone(\chi^{-1}(s))\right)\right)_j\right|
&= \left|\frac{u^j(G)}{k} - u^j(A_s)\right| \\
&= \left|\frac{1}{k} \sum_{i \in [k]} (u^j(A_i) - u^j(A_s))\right| \\
&\leq \frac{1}{k} \sum_{i \in [k]} \left|u^j(A_i) - u^j(A_s)\right| \\
&\overset{\eqref{eq:bounded-diff}}{\leq} \Delta - 1,
\end{align*}
where the first inequality follows from the triangle inequality.

Applying this for all $j \in [n]$, we have $\disc(\bA,k) \leq \Delta - 1$, contradicting our assumption that $\disc(\bA,k) > \Delta - 1$.
\end{proof}

\subsection{Upper Bound}

We next prove our upper bound (\Cref{thm:ub-main}), which turns out to be more challenging than the lower bound. To demonstrate this, let us consider using the ``reverse'' of the reduction in the proof of \Cref{thm:lb-direct}; specifically, suppose we create one row for each agent corresponding to her utilities. 
The discrepancy bound ensures that we can divide the goods into $k$ bundles so that each agent's \emph{utilities} for the $k$ bundles are not too different. 
However, this does not translate into any bound on the \emph{number} of goods necessary in the relaxation of any of the fairness notions, since it is possible that an agent has a tiny utility for every good in some bundle.

To tackle this issue, we must also ensure that each agent has some ``large'' (i.e., valuable) goods in every bundle. 
To this end, we divide the set of goods (with respect to each agent) into the set of large goods and the set of ``small'' goods. 
We create one row as before, but only for the small goods; this is to ensure that the utilities of the agent for the small goods do not differ by much between bundles. 
Additionally, we create a row corresponding to the large goods, which ensures that the agent has a non-trivial number of large goods in each bundle. 
When choosing the size of the set of large goods appropriately, this gives us the desired bound.
We formalize this intuition in the proof below.

\begin{proof}[Proof of~\Cref{thm:ub-main}]
Let $T := \lceil \exdisc(2n, k) \rceil$. 
For every agent $a^j$, let $S^j_{\text{large}} \subseteq G$ denote the set of $L := \min\{m, 3Tk\}$ goods that the agent values the most (ties broken arbitrarily), and let $p^j := \min_{\ell \in S^j_{\text{large}}} u^j(\ell)$. 
We define $\by^{j}$ as the $m$-dimensional indicator vector of $S^j_{\text{large}}$, i.e., $\by^{j} = \bone(S^j_{\text{large}})$.
We also define $\bz^j$ as the utility vector of the goods outside of $S^j_{\text{large}}$, scaled by $1 / p^j$, i.e.,
\begin{align*}
\bz^j_\ell =
\begin{cases}
u^j(\ell) / p^j & \text{ if } \ell \notin S^j_{\text{large}}; \\
0 & \text{ otherwise.}
\end{cases}
\end{align*}
By our choice of $S^j_{\text{large}}$ and $p^j$, we have $\bz^j \in [0, 1]^m$. (We use the convention $0/0 = 0$, i.e., $\bz^j$ is the all-zero vector when $p^j = 0$.)

We define our matrix $\bA = \begin{bmatrix}
\by^1
\cdots
\by^n
\bz^1
\cdots
\bz^n
\end{bmatrix}^T$; note that
$\bA \in [0, 1]^{2n \times m}$.
From the definition of $\exdisc$, there exists a coloring $\chi: [m] \to [k]$ such that
\begin{align} \label{eq:disc-bound}
\left\|\bA\left(\frac{1}{k} \cdot \bone - \bone(\chi^{-1}(i))\right)\right\|_{\infty} \leq T.
\end{align}
for all $i \in [k]$.
We pick our allocation $(A_1, \dots, A_k)$ according to $\chi$, that is, $A_i = \chi^{-1}(i)$ for all $i \in [k]$.

Next, we argue that for every pair $i, i' \in [k]$ and every agent $a^j$, there exists $B \subseteq A_{i'}$ of size at most $4T$ such that $u^j(A_i) \geq u^j(A_{i'} \setminus B)$; this suffices to finish the proof.
From \eqref{eq:disc-bound} and the definition of $\bA$, we have
\begin{align*}
T \geq \left|\left<\by^{j}, \left(\frac{1}{k} \cdot \bone - \bone(A_i)\right)\right>\right| = \left|\frac{L}{k} - |A_i \cap S^j_{\text{large}}|\right|
\end{align*}
and
\begin{align*} 
T 
&\geq \left|\left<\bz^j, \left(\frac{1}{k} \cdot \bone - \bone(A_i)\right)\right>\right| = \frac{1}{p^j} \cdot \left|\frac{1}{k} \cdot u^j(G \setminus S^j_{\text{large}}) - u^j(A_i \setminus S^j_{\text{large}})\right|.
\end{align*}
Rearranging these, we have
\begin{align} \label{eq:large-balanced-1}
\left|\frac{L}{k} - |A_i \cap S^j_{\text{large}}|\right| \leq T
\end{align}
and
\begin{align} \label{eq:small-balanced-1}
\left|\frac{1}{k} \cdot u^j(G \setminus S^j_{\text{large}}) - u^j(A_i \setminus S^j_{\text{large}})\right| \leq p^j \cdot T.
\end{align}
An analogous argument on bundle $A_{i'}$ implies
\begin{align} \label{eq:large-balanced-2}
\left|\frac{L}{k} - |A_{i'} \cap S^j_{\text{large}}|\right| \leq T
\end{align}
and
\begin{align} \label{eq:small-balanced-2}
\left|\frac{1}{k} \cdot u^j(G \setminus S^j_{\text{large}}) - u^j(A_{i'} \setminus S^j_{\text{large}})\right| \leq p^j \cdot T.
\end{align}

Let $B := A_{i'} \cap S^j_{\text{large}}$. 
By \eqref{eq:large-balanced-2}, we have $|B| \leq L/k + T \leq 4T$. Now, if $m \leq 3Tk$, then we have $L = m$ and $A_{i'} \setminus B = \emptyset$. Thus, $u^j(A_i) \geq u^j(A_{i'} \setminus B)$ trivially holds in this case.

Next, consider the case $m > 3Tk$, so $L = 3Tk$. In this case, we may bound $u^j(A_i) - u^j(A_{i'} \setminus B)$ as follows:
\begin{align*}
u^j(A_i) - u^j(A_{i'} \setminus B) 
&= u^j(A_i \cap S^j_{\text{large}}) + \left(u^j(A_i \setminus S^j_{\text{large}}) - u^j(A_{i'} \setminus S^j_{\text{large}})\right) \\
&\geq p^j \cdot |A_i \cap S^j_{\text{large}}| + \left(u^j(A_i \setminus S^j_{\text{large}}) - u^j(A_{i'} \setminus S^j_{\text{large}})\right) \\
&\overset{\eqref{eq:large-balanced-1}}{\geq}  p^j \left(\frac{L}{k} - T\right) + \left(u^j(A_i \setminus S^j_{\text{large}}) - u^j(A_{i'} \setminus S^j_{\text{large}})\right) \\
&\hspace{-1.9mm}\overset{\eqref{eq:small-balanced-1},\eqref{eq:small-balanced-2}}{\geq} p^j \left(\frac{L}{k} - T\right) - 2p^j \cdot T \\
&\geq 0,
\end{align*}
where the first inequality follows from our definition of $p^j$ and the last inequality follows from $L = 3Tk$.
This concludes the proof.
\end{proof}

\section{Lower Bounds From Weighted Discrepancy}
\label{sec:lower-weighted}

In this section, we prove a lower bound on $\cprop$ for $k$ groups via weighted discrepancy:

\begin{theorem} \label{thm:lb-prop}
For any $n', k \in \N$, we have 
\[\cpropemph(2n', 1, \dots, 1) \geq \lceil \exwdisc_{1/k}(n') / k \rceil.\]
\end{theorem}

Combined with \Cref{prop:fairness-relation}, this gives a similar lower bound for $\cef$ (again, the left-hand side contains $k-1$ $1$'s):
\begin{corollary} \label{cor:lb-ef}
For any $n', k \in \N$, we have 
\[\cefemph(2n', 1, \dots, 1) \geq \lceil \exwdisc_{1/k}(n') / k \rceil.\]
\end{corollary}

These two results, together with the lower bound on $\exwdisc_p$ (\Cref{prop:wdisc-lb}) and the observation that removing agents does not increase the value of $\cef$ or $\cprop$, yield the lower bound parts of \Cref{thm:main-ef,thm:main-prop}.
More specifically, we claim that for any $n_1,\dots,n_k$, we have
\begin{equation} \label{eq:remove-agents}
\cef(\max\{n_1, \dots, n_k\}, 1, 1, \dots, 1) \le \cef(n_1, \dots, n_k).
\end{equation}
To see why this claim holds, first note that due to symmetry, we may assume that $n_1 = \max\{n_1, \dots, n_k\}$. 
Note also that if an allocation satisfies a fairness notion for a certain instance, it still satisfies the same fairness notion when we arbitrarily remove some agent(s) from some group(s).
Hence, given any instance with group sizes $n_1, 1, \dots, 1$, we may add agents with arbitrary utilities so that the group sizes become $n_1, \dots, n_k$ and apply the bound for the latter case.
This implies \eqref{eq:remove-agents},
whereupon we can apply \Cref{cor:lb-ef} and \Cref{prop:wdisc-lb} to derive \Cref{thm:main-ef}. 
A similar argument can be used to derive \Cref{thm:main-prop} from \Cref{thm:lb-prop}.

Before we proceed to prove \Cref{thm:lb-prop}, we describe the intuition behind it.
We construct the utility functions of the agents so that if $(A_1, \dots, A_k)$ is proportional up to a small number of goods, then $\bx := \bone(A_1)$ gives us a small $1/k$-weighted discrepancy. 
Similarly to the proof of \Cref{thm:lb-direct}, we start by creating $n'$ agents in the first group where each agent's utilities correspond to a row of $\bA$. 
This yields a lower bound on each entry of $\bA \cdot \bone(A_1)$. 
To get an upper bound, we simply create a ``conjugate'' of each of these agents in the first group, i.e., the conjugate utility of each good is simply $1$ minus the original utility. 
However, this construction alone is not sufficient---for example, it is still possible to assign all goods to $A_1$. 
To avoid this, we create one agent in each of the remaining groups with the same utility for all goods. 
This ensures that $A_1$ has size roughly $m/k$, which turns out to be sufficient for bounding the $1/k$-weighted discrepancy (using $\bx := \bone(A_1)$).

\begin{proof}[Proof of \Cref{thm:lb-prop}]
Let $\Delta := \lceil \exwdisc_{1/k}(n') / k \rceil$. Note that $k(\Delta - 1) < \exwdisc_{1/k}(n')$. 
Thus, there exists $m\in\N$ and $\bA \in [0, 1]^{n' \times m}$ such that $\wdisc_{1/k}(\bA) > k(\Delta - 1)$. 
For each $\ell\in[m]$, we define the utilities of the agents by
\begin{align*}
u^{(i, 1)}(\ell) = 1
\end{align*}
for all $i \in \{2, \dots, k\}$, and
\begin{align*}
u^{(1, j)}(\ell) = \bA_{j, \ell}, \qquad u^{(1, j + n')}(\ell) = 1 - \bA_{j, \ell}
\end{align*}
for all $j \in [n']$.

Suppose for the sake of contradiction that there is an allocation $(A_1, \dots, A_k)$ that is proportional up to $\Delta - 1$ goods.
Consider agent $a^{(i, 1)}$ for all $i \in \{2, \dots, k\}$. 
The proportionality up to $\Delta - 1$ goods of the agent implies that $|A_i| \geq m / k - (\Delta - 1)$. 
From this, we have
\begin{align} \label{eq:first-bundle-size}
|A_1| \leq \frac{m}{k} + (k - 1)(\Delta - 1).
\end{align}

Next, consider agent $a^{(1, j)}$ for $j\in[n']$.
Since the allocation is proportional up to $\Delta - 1$ goods and each good has value at most $1$ to the agent, we have
\begin{align*}
u^{(1, j)}(A_1) &\geq \frac{u^{(1, j)}(G)}{k} - (\Delta - 1),
\end{align*}
or equivalently,
\begin{align}
\label{eq:delta-positive}
\Delta - 1 &\geq \frac{u^{(1, j)}(G)}{k} - u^{(1, j)}(A_1).
\end{align}
Similarly, if we instead consider agent $a^{(1, j + n')}$ for $j\in[n']$, we have
\begin{align*}
\Delta - 1 &\geq \frac{u^{(1, j + n')}(G)}{k} - u^{(1, j + n')}(A_1) \\
&= \frac{m - u^{(1, j)}(G)}{k} - (|A_1| - u^{(1, j)}(A_1)) \\
&= \left(\frac{m}{k} - |A_1|\right) - \left(\frac{u^{(1, j)}(G)}{k} - u^{(1, j)}(A_1)\right) \\
&\overset{\eqref{eq:first-bundle-size}}{\geq} -(k - 1)(\Delta - 1) - \left(\frac{u^{(1, j)}(G)}{k} - u^{(1, j)}(A_1)\right),
\end{align*}
which means that
\begin{align}
\label{eq:delta-negative}
k(\Delta - 1) \ge -\left(\frac{u^{(1, j)}(G)}{k} - u^{(1, j)}(A_1)\right).
\end{align}

From \eqref{eq:delta-positive} and \eqref{eq:delta-negative}, we can conclude that
\begin{align*}
k(\Delta - 1) &\geq \left|\frac{u^{(1, j)}(G)}{k} - u^{(1, j)}(A_1)\right| = \left|\left(\bA\left(\frac{1}{k}\cdot\bone - \bone(A_1)\right)\right)_j\right|.
\end{align*}

Applying this for all $j \in [n']$, we have $\wdisc_{1/k}(\bA) \leq k(\Delta - 1)$, contradicting our assumption that $\wdisc_{1/k}(\bA) > k(\Delta - 1)$.
\end{proof}

\section{Computational Complexity}
\label{sec:complexity}

Since efficient algorithms matching the bound in \Cref{lem:disc-ub-multi} are known~\citep{Bansal10,LovettMe15,LevyRR17}\footnote{Please refer to Appendix~\ref{app:discrepancy-algo} for more details.} and all of our upper bounds are obtained by polynomial-time reductions to this bound, it immediately follows that given the goods and the agents' utilities for them, we can efficiently find an allocation matching the upper bounds in our main theorems (Theorems~\ref{thm:main-ef}--\ref{thm:main-cd}). In summary, we have:
\begin{corollary} \label{cor:efficient-algo}
There exists a deterministic polynomial-time algorithm that can compute a consensus $1/k$-division (or an envy-free/proportional allocation) up to $O(\sqrt{n})$ goods.
\end{corollary}

In light of the above corollary, a natural question is whether we can improve on this $O(\sqrt{n})$ bound if we know that, say, an unknown ``fully fair'' division exists for a given instance. 
For example, provided that there is a consensus $1/k$-division in that instance, can we efficiently find an allocation that beats the upper bounds in \Cref{cor:efficient-algo}? 

A similar question has been asked in the context of discrepancy theory; for the bound in \Cref{lem:disc-ub}, the answer was shown to be negative~\citep{CharikarNeNi11}, i.e., even when $\bA$ has discrepancy zero, it is NP-hard to find $\bx$ achieving discrepancy $o(\sqrt{n})$. In this section, we extend this hardness to the setting of fair division, as stated below.

\begin{theorem} \label{thm:hardness-fair-division}
For any constant $k \in \N \setminus \{1\}$, there exists a constant $\eps_k > 0$ such that it is NP-hard, given $m$ goods and $k$ groups, each containing $n'$ agents with additive utilities, to distinguish between the following two cases:
\begin{itemize}
\item (YES) There exists a consensus $1/k$-division;
\item (NO) No allocation is proportional up to $\eps_k \sqrt{n'}$ goods.
\end{itemize}
\end{theorem}

As a consequence, when $k$ is constant, we cannot asymptotically improve upon the bound in \Cref{cor:efficient-algo} even when we are promised that a consensus $1/k$-division exists (assuming P $\ne$ NP).
Indeed, if there exists a polynomial-time algorithm that improves this bound, then given an instance belonging to one of the two cases of \Cref{thm:hardness-fair-division}, we can run this algorithm on it.
If the instance admits a consensus $1/k$-division, the output of the algorithm will be a consensus $1/k$-division up to $o(\sqrt{n})$ goods, while if no allocation is proportional up to $\varepsilon_k\sqrt{n}$ goods, then the output will be a consensus $1/k$-division up to $\Omega(\sqrt{n})$ goods (or the algorithm does not output any allocation).
Hence, we would be able to distinguish between the two cases of \Cref{thm:hardness-fair-division} in polynomial time; however, this is impossible by \Cref{thm:hardness-fair-division} unless P $=$ NP.
Note also that since consensus $1/k$-division is the strongest notion and proportionality the weakest (see \Cref{prop:fairness-relation}), \Cref{thm:hardness-fair-division} is the ``strongest possible''.

The remainder of this section is devoted to the proof of \Cref{thm:hardness-fair-division} and is organized as follows. 
First, in Section~\ref{subsec:hardness-multicolor}, we generalize the result of~\cite{CharikarNeNi11} to the setting of multi-color discrepancy.
Then, in Section~\ref{sec:hardness-multicolor-to-fair-div}, we reduce from the hardness of multi-color discrepancy to the hardness of (almost) fair division.

\subsection{Hardness of Multi-Color Discrepancy}
\label{subsec:hardness-multicolor}

We start by proving the following hardness result, which is a generalization of Charikar et al.'s result  for $k = 2$ to arbitrary~$k$.

\begin{lemma} \label{lem:hardness-multicolor-disc}
For any $k \in \N \setminus \{1\}$, there exists $\delta > 0$ such that the following holds: It is NP-hard, given a matrix $\bA \in [0, 1]^{n \times m}$, to distinguish between the following two cases:
\begin{itemize}
\item (YES) $\disc(\bA, k) = 0$;
\item (NO) $\disc(\bA, k) > \delta\sqrt{n}$.
\end{itemize}
\end{lemma}

Our proof follows the framework of~\cite{CharikarNeNi11}: In the first step, we prove a hardness for when the NO case only guarantees that the discrepancy is at least constant, but with some structure on the given matrix. 
Then, in the second step, we compose the instance with a Hadamard matrix that amplifies the NO case to the desired $\Omega(\sqrt{n})$ bound.

\subsubsection*{Step I: Constant Factor Hardness}

In this step, we will prove the following hardness.

\begin{lemma} \label{lem:hardness-multicolor-l2}
For any $k \in \N \setminus \{1\}$, there exist $D \in \N$ and $\eps > 0$ such that the following holds: It is NP-hard, given a matrix $\bB \in [0, 1]^{n \times m}$ where each column of $\bB$ has $\ell_1$-norm at most $D$ and each row of $\bB$ has $\ell_1$-norm at most $k$, to distinguish between the following two cases:
\begin{itemize}
\item (YES) $\disc(\bB, k) = 0$;
\item (NO) For any coloring $\chi: [m] \to [k]$, we have $\max_{s \in [k]} \left\|\bB\left(\frac{1}{k} \cdot \bone - \bone(\chi^{-1}(s))\right)\right\|_2 > \eps \sqrt{n}$.
\end{itemize}
\end{lemma}

Note that the condition in the NO case only ensures that $\left\|\bB\left(\frac{1}{k} \cdot \bone - \bone(\chi^{-1}(s))\right)\right\|_\infty \geq \Omega(1)$. However, we require the $\ell_2$ version instead of the $\ell_\infty$ version because the former will be used in the subsequent step when we compose our instance with a Hadamard matrix.

Following~\cite{CharikarNeNi11}, we reduce from the \setsplit\ problem.
In this problem, we are given subsets $S_1, \dots, S_N \subseteq [M]$, each of size exactly $4$, and the goal is to find a subset $T \subseteq [M]$ such that $|S_i \cap T| = 2$ for each subset $S_i$. 
Such a set $S_i$ is said to be \emph{split}; otherwise, the set is said to be \emph{unsplit}. The \setsplit\ problem is known to be hard to approximate:
\begin{lemma}[\citep{Guruswami03,CharikarGuWi05}] \label{lem:setsplit-hardness}
There exist $d \in \N$ and $\delta > 0$ such that it is NP-hard, given a \setsplit\ instance where each element appears in at most $d$ subsets, to distinguish between the following two cases:
\begin{itemize}
\item (YES) There exists $T \subseteq [M]$ that makes every subset split;
\item (NO) Any subset $T \subseteq [M]$ leaves more than $\delta N$ subsets unsplit.
\end{itemize}
\end{lemma}

The \setsplit\ instance naturally corresponds to the ($2$-color) discrepancy framework. 
In particular, by viewing the set system $S_1, \dots, S_N \subseteq [M]$ as an indicator $\bB \in \{0, 1\}^{N \times M}$, one can verify that \Cref{lem:setsplit-hardness} immediately gives \Cref{lem:hardness-multicolor-l2} for the case $k = 2$; this is exactly the first step of~\cite{CharikarNeNi11}.

To extend this argument to the multi-color case, we have to add extra columns to the matrix so that these columns can be used for the other $k - 2$ colors.
The main challenge here is in the NO case: if we create the columns carelessly, it might be possible to ``mix'' these new columns with the old ones so that the discrepancy becomes small, even when there is no good solution to the \setsplit\ instance. 
To overcome this issue, we create multiple copies of the original rows and combine them with the new columns. 
Roughly speaking, this means that, if we mix many new columns and many old ones, then there will be a large amount of ``collision'' which makes the corresponding coordinate too large, resulting in a large discrepancy. 
To formalize this gadget, we will use a sufficiently strong expander graph; its properties that we need are summarized below.

\begin{lemma}[{e.g., \citep{Cohen16}}] \label{lem:expander}
For any $\gamma > 0$, there exists $d' \in \N$ and $\beta > 0$, and an algorithm that takes in $N \in \N$, runs in $\poly(N, d')$ time, and produces a $d'$-regular bipartite multigraph $G = (L, R, E)$ where $L = R = [N]$ and, for any subsets $U \subseteq L, V \subseteq R$ each of size at least $\gamma N$, the number of edges between $U, V$ is more than $\beta d' N$.
\end{lemma}

\begin{proof}[Proof of \Cref{lem:hardness-multicolor-l2}]
Fix $k \in \N \setminus \{1\}$.
Let $S_1, \dots, S_N \subseteq [M]$ where each element appears in at most $d$ subsets be an instance of \setsplit\ from \Cref{lem:setsplit-hardness}. 
Assume without loss of generality that $\delta \le 8$.
We create a $d'$-regular multigraph $G = (L, R, E)$ as in \Cref{lem:expander} with $L = R = [N]$ and $\gamma = \frac{\delta}{8dk}$, and let the $i$th edge in $E$ be $e_i = (u_i, v_i)$.

We then construct the matrix $\bC \in [0, 1]^{|E| \times M}$ by
\begin{align*}
\bC_{i, j} = 
\begin{cases}
1/2 & \text{ if } j \in S_{u_i}; \\
0 & \text{ otherwise,}
\end{cases}
\end{align*}
and $\bD \in [0, 1]^{|E| \times N}$ by
\begin{align*}
\bD_{i, j} = 
\begin{cases}
1 & \text{ if } j = v_i; \\
0 & \text{ otherwise.}
\end{cases}
\end{align*}
Finally, let 
\[\bB = \begin{bmatrix}\bC &\bD &\cdots &\bD\end{bmatrix} \in [0, 1]^{|E| \times (M + (k - 2)N)}\]
where $\bD$ is repeated $k - 2$ times. (Note that we have $n = |E| = Nd'$ and $m = M + (k - 2)N$.) 
Since each vertex of $G$ has degree $d'$ and each element appears in at most $d$ subsets, each column of $\bB$ has $\ell_1$ norm at most $D := \max\{d'd/2, d'\}$, where the first term follows from the definition of $\bC$ and the second term from that of $\bD$. Furthermore, each row of $\bB$ has $\ell_1$ norm exactly $4 \cdot 1/2 + (k - 2) \cdot 1 = k$.

The reduction clearly runs in polynomial time. We will next prove its correctness with $\eps := \min\{\sqrt{\beta}/2, \sqrt{\delta/8}\}$, where $\delta$ is as in \Cref{lem:setsplit-hardness} and $\beta$ is as in \Cref{lem:expander}.

\paragraph{(YES)} Suppose that there exists $T \subseteq [M]$ that splits all subsets. 
We define the desired coloring by
\begin{itemize}
\item $\chi^{-1}(1) = T$;
\item $\chi^{-1}(2) = [M] \setminus T$;
\item $\chi^{-1}(\ell) = \{M + (\ell - 3)N + 1, \dots, M + (\ell - 2)N\}$ for all $\ell \in \{3, \dots, k\}$.
\end{itemize}

We claim that $\bB\left(\frac{1}{k} \cdot \bone - \bone(\chi^{-1}(s))\right) = \bzero$ for all $s \in [k]$; this claim implies that $\disc(\bB, k) = 0$. 
Observe that $\bB\frac{1}{k} \cdot \bone = \bone$ because each row of $\bB$ contains exactly four copies of $1/2$ and $(k - 2)$ copies of $1$. 
Hence, we are left to show that $\bB\bone(\chi^{-1}(s)) = \bone$ for all $s \in [k]$. 
To do this, let us consider three cases:
\begin{itemize}
\item Case I: $s = 1$. In this case, we have $\bB\bone(\chi^{-1}(s)) = \bC\bone(T)$. For each $i \in [|E|]$, we have
\begin{align*}
(\bC\bone(T))_i = \frac{1}{2} \cdot |S_{u_i} \cap T| = 1,
\end{align*}
where the second equality follows from the assumption that $T$ splits $S_{u_i}$. 
Hence, we have $\bB\bone(\chi^{-1}(s)) = \bone$.
\item Case II: $s = 2$. This case is similar to the Case~I since $[M] \setminus T$ also splits every subset.
\item Case III: $s \in \{3, \dots, k\}$. In this case, we have $\bB\bone(\chi^{-1}(s)) = \bD\bone = \bone$.
\end{itemize}

\paragraph{(NO)} Suppose contrapositively that there exists a coloring $\chi: [m] \to [k]$ such that $\left\|\bB\left(\frac{1}{k} \cdot \bone - \bone(\chi^{-1}(s))\right)\right\|_2 \leq \eps \sqrt{n}$ for all $s \in [k]$. We may assume without loss of generality that 
\[
|\chi^{-1}(1) \cap [M]| \geq \max\{|\chi^{-1}(2) \cap [M]|, \dots, |\chi^{-1}(k) \cap [M]|\}.
\]
For notational convenience, let $P = \chi^{-1}(1)$ and $T = P \cap [M]$. Note that our assumption implies that $|T| \geq M/k$.

 We will next bound the size of $|P \cap \{M + (\ell - 3)N + 1, \dots, M + (\ell - 2)N\}|$ for each $\ell \in \{3, \dots, k\}$.

\begin{claim} \label{claim:small-intersection}
For any $\ell \in \{3, \dots, k\}$, we have $|P \cap \{M + (\ell - 3)N + 1, \dots, M + (\ell - 2)N\}| < \gamma N$.
\end{claim}

\begin{proof}[Proof of Claim]
Fix any $\ell \in \{3, \dots, k\}$.
Let $U = \{u \in [N] \mid S_u \cap T \ne \emptyset\}$, and $V = \{v \in [N] \mid M + (\ell - 3)N + v \in P\}$. We have $$|U| \geq \frac{|T|}{4} \geq \frac{M}{4k} \geq \frac{N}{dk} \geq \gamma N,$$
where the third inequality holds since each element of $[M]$ appears in at most $d$ subsets $S_i$ and each of these subsets contains exactly four elements.

Notice that, for every edge $(u_i, v_i) \in U \times V$, it holds that \begin{align*}
(\bB\bone(P))_i 
&\geq \frac{1}{2} \cdot |T \cap S_{u_i}| + (\bone(P))_{M + (\ell - 3)N + v_i} \geq \frac{1}{2} + 1 = \frac{3}{2}.
\end{align*}
From the guarantee that $\left\|\bB\left(\frac{1}{k} \cdot \bone - \bone(P)\right)\right\|_2 \leq \eps \sqrt{n}$ and since $\bB\frac{1}{k} \cdot \bone =  \bone$, we can conclude that 
\begin{align*}
|(U \times V) \cap E| \leq \frac{\eps^2 n}{(1/2)^2} = 4\eps^2|E| = 4\eps^2 d' N \leq \beta d' N.
\end{align*}
From the above inequality, the fact that $|U| \geq \gamma N$, and the guarantee of \Cref{lem:expander}, we can conclude that $|V| < \gamma N$, as desired.
\end{proof}

The above claim implies that $|P \setminus T| < \gamma(k - 2)N$. Since $G$ is $d'$-regular (in particular, each $j\in[N]$ appears as $v_i$ exactly $d'$ times), this in turn yields
\begin{align} 
\|\bB(\bone(P) - \bone(T))\|_0 
&= \|\bB\bone(P\setminus T)\|_0 \nonumber \\
&\leq d' \gamma (k - 2)N = \gamma(k - 2)|E|.
\label{eq:difference}
\end{align}
To put it another way, $\bB\bone(P)$ and $\bB\bone(T)$ differ in at most $\gamma(k - 2)|E|$ coordinates. Recall that 
\begin{equation} \label{eq:B1P}
\eps \sqrt{n} \geq \left\|\bB\left(\frac{1}{k} \cdot \bone - \bone(P)\right)\right\|_2 = \|\bone - \bB\bone(P)\|_2.
\end{equation}
Putting these two together and recalling that $|E| = n$, we have
\begin{align*}
\|\bone - \bB\bone(T)\|_2^2 &= \sum_{i\in[|E|]} (1 - (\bB\bone(T))_i)^2 \\
&\overset{\eqref{eq:B1P}}{\leq} \eps^2 n + \sum_{i \in [|E|]} \left((1 - (\bB\bone(T))_i)^2 - (1 - (\bB\bone(P))_i)^2\right) \\
&\leq \eps^2 n + \sum_{i \in [|E|] \atop (\bB\bone(T))_i \ne (\bB\bone(P))_i} (1 - (\bB\bone(T))_i)^2 \\
&\leq \eps^2 n + \|\bB(\bone(P) - \bone(T))\|_0 \\
&\overset{\eqref{eq:difference}}{\leq} \eps^2 n + \gamma(k - 2)|E|,
\end{align*}
where the third inequality holds because $(\bB\bone(T))_i\in[0,2]$ for all $i\in[|E|]$.

Finally, notice that each subset $S_i$ unsplit by $T$ contributes at least $d'/4$ to $\|\bone - \bB\bone(T)\|_2^2$, because the set $S_i$ corresponds to $d'$ rows and each row has absolute value at least $1/2$. 
This means that the number of subsets unsplit by $T$ is at most
\begin{align*}
\frac{4}{d'}\left(\eps^2 n + \gamma(k - 2)|E|\right)
&= 4\left(\eps^2 + \gamma(k - 2)\right) N \\
&\leq 4\left(\frac{\delta}{8} + \frac{\delta}{8}\right) N = \delta N,
\end{align*}
where the inequality follows from our definitions of $\eps$ and $\delta$.
This concludes our proof of \Cref{lem:hardness-multicolor-l2}.
\end{proof}

\subsubsection*{Step II: Gap Amplification via Hadamard Matrix}

Next, we follow the proof of~\citet{CharikarNeNi11} by composing the hard instance in the previous subsection with a Hadamard matrix, which allows us to prove \Cref{lem:hardness-multicolor-disc}. In fact, our proof is slightly simpler than theirs, since they would like the resulting matrix $\bA$ to have Boolean entries, whereas we only need the entries to belong to $[0, 1]$.

\begin{proof}[Proof of \Cref{lem:hardness-multicolor-disc}]
Let $\bB \in [0, 1]^{n \times m}$ be the instance from~\Cref{lem:hardness-multicolor-l2} and $D, \eps$ be as in~\Cref{lem:hardness-multicolor-l2}.

We may assume that $n$ is a power of two; otherwise, we add all-zero rows to $\bB$ until $n$ becomes a power of two. Let $\bH \in \{-1, +1\}^{n\times n}$ be a Hadamard matrix of order $n$ and $\bW \in \{0, 1\}^{n\times n}$ be as in \Cref{lem:hadamard-discrepancy}. We simply let $\bA := \frac{1}{D} \cdot \bW\bB$. Since each column of $\bB$ has $\ell_1$-norm at most $D$ and $\bW$'s entries are $0$ or $1$, we have $\bA \in [0, 1]^{n \times m}$.

\paragraph{(YES)} If $\disc(\bB, k) = 0$, then there exists a coloring $\chi: [m] \to [k]$ such that, for all $s \in [k]$, 
\[\bB\left(\frac{1}{k} \cdot \bone - \bone(\chi^{-1}(s))\right) = \bzero.\]
This implies that $\bA\left(\frac{1}{k} \cdot \bone - \bone(\chi^{-1}(s))\right) = \bzero$, meaning that $\disc(\bA, k) = 0$.

\paragraph{(NO)} Suppose that for any coloring $\chi: [m] \to [k]$, it holds that \[\left\|\bB\left(\frac{1}{k} \cdot \bone - \bone(\chi^{-1}(s))\right)\right\|_2 > \eps \sqrt{n}\] for some $s \in [k]$. We have
\begin{align*}
&\left\|\bA\left(\frac{1}{k} \cdot \bone - \bone(\chi^{-1}(s))\right)\right\|_2^2 = \frac{1}{D^2}\left\|\bW\bB\left(\frac{1}{k} \cdot \bone - \bone(\chi^{-1}(s))\right)\right\|_2^2.
\end{align*}
Next, notice that since the first row of $\bB$ has $\ell_1$ norm at most $k$ and all coordinates of $\frac{1}{k} \cdot \bone - \bone(\chi^{-1}(s))$ have absolute value at most $1$, the first coordinate of $\bB\left(\frac{1}{k} \cdot \bone - \bone(\chi^{-1}(s))\right)$ has absolute value at most $k$. Using this fact and \Cref{lem:hadamard-discrepancy}, we can conclude that
\begin{align*}
\left\|\bA\left(\frac{1}{k} \cdot \bone - \bone(\chi^{-1}(s))\right)\right\|_2^2 
&\geq \frac{n}{D^2}\left(\left\|\bB\left(\frac{1}{k} \cdot \bone - \bone(\chi^{-1}(s))\right)\right\|_2^2 - k^2\right) \\
&> \frac{n}{D^2} \left(\eps^2 n - k^2\right),
\end{align*}
which is at least $(0.5\eps^2/D^2)n^2$ for any sufficiently large $n$.
This implies that \[\left\|\bA\left(\frac{1}{k} \cdot \bone - \bone(\chi^{-1}(s))\right)\right\|_{\infty} > \delta \sqrt{n}\] where $\delta = \sqrt{0.5\eps^2/D^2}$. Thus, we have $\disc(\bA, k) > \delta\sqrt{n}$.
\end{proof}

\subsection{From Multi-Color Discrepancy to Fair Division}
\label{sec:hardness-multicolor-to-fair-div}

We now prove \Cref{thm:hardness-fair-division} by reducing multi-color discrepancy to (almost) fair division. This reduction is similar to the ones we used earlier to establish our worst-case bounds.

\begin{proof}[Proof of \Cref{thm:hardness-fair-division}]
Let $\bA \in [0, 1]^{n \times m}$ be the instance from \Cref{lem:hardness-multicolor-disc} and $\delta$ be as in the lemma. We choose $\eps_k := \delta / (k - 1)$.

We construct one agent in each group corresponding to a row of $\bA$. Specifically, the utility of agent $a^{(i, j)}$ for good~$\ell$ is $u^{(i, j)}(\ell) = \bA_{j, \ell}$ for all $i \in [k]$, $j \in [n]$, and $\ell \in [m]$.

\paragraph{(YES)} Suppose that $\disc(\bA, k) = 0$, that is, there exists a coloring $\chi: [m] \to [k]$ such that, for all $s \in [k]$, $\bA\left(\frac{1}{k} \cdot \bone - \bone(\chi^{-1}(s))\right) = \bzero$. 
Let $(A_1, \dots, A_k)$ be the allocation where $A_s = \chi^{-1}(s)$ for all $s \in [k]$. 
One can check that this allocation is a consensus $1/k$-division for every agent.

\paragraph{(NO)} Suppose contrapositively that there exists an allocation $(A_1, \dots, A_k)$ that is proportional up to $\eps_k \sqrt{n}$ goods. 
Consider the coloring $\chi: [m] \to [k]$ defined by $\chi^{-1}(s) = A_s$ for all $s \in [k]$. 
Since the allocation is proportional up to $\eps_k \sqrt{n}$ goods for agent $a^{(i, j)}$, we have
\begin{align*}
\eps_k\sqrt{n} 
&\geq \frac{u^{(i, j)}(G)}{k} - u^{(i, j)}(A_i) = \left(\bA\left(\frac{1}{k} \cdot \bone - \bone(\chi^{-1}(i))\right)\right)_j.
\end{align*}
Summing this over all indices $i$ besides a fixed $i^*$, we get
\begin{align*}
\eps_k (k - 1) \sqrt{n} 
&\geq \left(\bA\left(\frac{k - 1}{k} \cdot \bone - \left(\bone - \bone(\chi^{-1}(i^*))\right)\right)\right)_j \\
&= -\left(\bA\left(\frac{1}{k} \cdot \bone - \bone(\chi^{-1}(i^*))\right)\right)_j.
\end{align*}
Combining the two inequalities (the first for $i=i^*$) yields
\begin{align*}
\eps_k(k - 1)\sqrt{n} \geq \left|\left(\bA\left(\frac{1}{k} \cdot \bone - \bone(\chi^{-1}(i^*))\right)\right)_j\right|.
\end{align*}
Since this holds for all $i^*\in [k]$ and $j \in [n]$, we can conclude that $\disc(\bA, k) \leq \eps_k(k - 1)\sqrt{n} = \delta\sqrt{n}$.
\end{proof}

\section{Conclusion and Future Work}

In this paper, we have studied the allocation of indivisible goods to groups of agents using the standard fairness notions of envy-freeness, proportionality, and consensus $1/k$-division.
We presented bounds on the optimal relaxations of these notions that can be guaranteed for agents with additive valuations; all of the bounds are asymptotically tight when the number of groups is constant.
Our results imply that relatively strong fairness guarantees can be provided for all agents even when agents in the same group, who share the same set of resources, have highly differing preferences.
Moreover, we showed that computing allocations that improve upon these bounds is NP-hard even in instances where such allocations are known to exist.

Besides closing the gaps left by our work, an interesting direction for future work is to consider agents with arbitrary monotonic utilities.
Indeed, the techniques from discrepancy theory that we used crucially rely on the additivity assumption; so does the result of \citet{Alon87} that established the existence of a consensus $1/k$-division for divisible goods. 
Even in the case of prime numbers $k$, where a consensus $1/k$-division can be guaranteed for non-additive utilities \citep{FilosratsikasHoSo21},\footnote{See Theorem~6.5 in their extended version. For $k=2$, the existence of a consensus halving with non-additive utilities was shown by \citet{SimmonsSu03}.} it is unclear whether such a division can be rounded into a discrete allocation with a loss that is bounded only in terms of $n$.
Beyond the setting of our paper, one could also consider allocating a mixture of indivisible and divisible resources \citep{BeiLiLi21,BeiLiLu21,BhaskarSrVa21} or allowing groups to have different entitlements which can correspond to the group sizes \citep{FarhadiGhHa19,BabaioffEzFe21,ChakrabortySeSu22,SuksompongTe22} as well.

\section*{Acknowledgments}

This work was partially supported by the Singapore Ministry of Education under grant number MOE-T2EP20221-0001 and by an NUS Start-up Grant.
We would like to thank Paul Goldberg, Alexandros Hollender, and Ayumi Igarashi for interesting discussions and the anonymous reviewers of the 30th International Joint Conference on Artificial Intelligence (IJCAI 2021) and Theoretical Computer Science for valuable comments.

\bibliographystyle{plainnat}
\bibliography{main}

\appendix

\section{Proof Sketch of Constructive Version of Multi-Color Discrepancy}
\label{app:discrepancy-algo}

Since the work of \citet{DoerrSr03} did not explicitly discuss the computational efficiency of their bound (\Cref{lem:disc-ub-multi} in our paper), we will roughly sketch the proof of the efficient version of their bound; more specifically, we will outline the proof of the following theorem. We stress that this is already implicitly known and we merely include the arguments here for completeness.

\begin{theorem} \label{thm:multi-color-efficient-ub}
Given any $\bA \in [0, 1]^{m \times n}$ and $k \in \N$, there exists a deterministic polynomial-time algorithm that computes a coloring $\chi: [m] \to [k]$ such that
\begin{align*}
\left\|\bA\left(\frac{1}{k} \cdot \bone - \bone(\chi^{-1}(s))\right)\right\|_{\infty} \leq O(\sqrt{n})
\end{align*}
for all $s \in [k]$.
\end{theorem}

To derive this theorem, we first recall a result due to \citet{LevyRR17},\footnote{\citet{LovettMe15} earlier proved a similar result but their algorithm is randomized.} who gave an efficient algorithm for the case of two colors; in fact, their algorithm works for a more general notion of discrepancy called \emph{linear discrepancy}~\citep{LovaszSpVe86}, where there is a ``starting vector'' $\bw \in [0, 1]^n$ and the goal is to find a $2$-coloring mostly resembling this starting vector. When $n \geq m$, their algorithm implies the following.

\begin{theorem}[\citep{LevyRR17}]
Suppose that $n \geq m$. There exists a deterministic polynomial-time algorithm that, given any $\bA \in [0, 1]^{n \times m}$ and $\bw \in [0, 1]^m$, can find $\bx \in \{0, 1\}^m$ such that $\|\bA(\bw - \bx)\|_{\infty} \leq O(\sqrt{m \log(2n/m)})$.
\end{theorem}

Now the above bound is not quite what we want, due to its requirement that $n \geq m$. However, it turns out that there is an efficient algorithm that can always reduce the case $m > n$ to $m \leq n$~\cite[Theorem 12.3.1]{AlonSp00} for which the bound $O(\sqrt{m \log(2n/m)})$ becomes at most $O(\sqrt{n})$. Combining these two, we get:
\begin{corollary} \label{cor:efficient-linear-disc}
There exists a deterministic polynomial-time algorithm that, given any $\bA \in [0, 1]^{n \times m}$ and $\bw \in [0, 1]^m$, can find $\bx \in \{0, 1\}^m$ such that $\|\bA(\bw - \bx)\|_{\infty} \leq O(\sqrt{n})$.
\end{corollary}

We can now give a proof sketch of \Cref{thm:multi-color-efficient-ub} via an algorithm based on the work of \cite{DoerrSr03}.

\begin{proof}[Proof sketch of \Cref{thm:multi-color-efficient-ub}]
The algorithm is a recursive algorithm that works as follows:
\begin{itemize}
\item If $k = 1$, then halt and output the only possible coloring. 
\item Let $k_0 = \lfloor k / 2 \rfloor$ and $k_1 = k - k_0$, and let $\bw = (k_1/k) \cdot \bone$.
\item Use the algorithm from \Cref{cor:efficient-linear-disc} to find $\bx$ for $\bw$ as above (and the input $\bA$). 
\item For $i \in \{0, 1\}$:
\begin{itemize}
\item Let $\bA^{(i)}$ denote the submatrix of $\bA$ restricted to only columns $j$ with $\bx_j = i$.
\item Recursively run the algorithm on $\bA^{(i)}$ with $k_i$ colors to get a coloring $\chi_i$.
\end{itemize}
\item Output the coloring $\chi$ that results from concatenating $\chi_0$ and $\chi_1$, where we use disjoint sets of colors for the two colorings.
\end{itemize}

It is clear that this algorithm runs in polynomial time. Furthermore, similarly to Corollary 3.5 of \citet{DoerrSr03}, it can be verified that the guarantee in \Cref{cor:efficient-linear-disc} implies that $\left\|\bA\left(\frac{1}{k} \cdot \bone - \bone(\chi^{-1}(s))\right)\right\|_{\infty} \leq O(\sqrt{n})$ for all $s \in [k]$.
\end{proof}

Finally, we note that our reduction in the proof of \Cref{thm:ub-main} is efficient, so we immediately get the claimed polynomial-time algorithm for consensus $1/k$-division up to $O(\sqrt{n})$ goods. 
Furthermore, the reductions in \Cref{prop:fairness-relation} then yield polynomial-time algorithms for envy-freeness up to $O(\sqrt{n})$ goods and proportionality up to $O(\sqrt{n})$ goods.

\end{document}